\documentclass[12pt]{article}

\usepackage[margin=1in]{geometry}
\usepackage{amsmath,amsfonts,amssymb,amsthm}
\usepackage{enumitem}
\usepackage{booktabs,longtable,array}
\usepackage[round]{natbib}
\usepackage[hidelinks]{hyperref}
\usepackage{newpxtext,newpxmath}
\usepackage{microtype, centernot, csquotes, tikz}
\usetikzlibrary{shapes.geometric}
\newtheorem{theorem}{Theorem}
\newtheorem{lemma}{Lemma}
\newtheorem{proposition}{Proposition}

\newtheorem{remark}{Remark}
\newtheorem{definition}{Definition}

\newtheorem{example}{Example}

\newcommand{\lex}{\mathrm{lex}}
\newcommand{\R}{\mathbb{R}}
\newcommand{\I}{\mathbb{I}}
\newcommand{\N}{\mathbb{N}}

\newcommand{\calS}{\mathcal{S}}
\newcommand{\Z}{\mathbb{Z}}
\newcommand{\OmegaC}{\Omega}

\title{On the construction and representation of social welfare orders satisfying consequentialist equity axioms}
\author{Ram Sewak Dubey\thanks{%
Economics Department, Feliciano School of Business, Montclair State
University, Montclair, NJ 07043, USA; E-mail: dubeyr@montclair.edu} \and %
Giorgio Laguzzi\thanks{%
Universita~ del Piemonte Orientale, 15121 Alessandria, Italy; Email:
giorgio.laguzzi@uniupo.it}}
\date{\today }

\begin{document}

\maketitle
\begin{abstract}
In this paper we examine the constructive nature of social welfare orders on infinite utility streams $X=Y^{\mathbb{N}}$ satisfying Strong Equity, Hammond Equity, or the Pigou--Dalton transfer principle.
The constructive social welfare orders are described using lexicographic preference relations.
Social welfare orders satisfying Strong Equity, Hammond Equity, or the Pigou--Dalton transfer principle admit explicit descriptions when $Y(<)$ is well-ordered.
We also show that a real-valued social welfare function satisfying Strong Equity exists if and only if $Y$ is countable.
We describe restrictions on the domain $Y$ under which the existence of social welfare orders satisfying the aforementioned equity axioms entails the existence of a non-Ramsey collection.
For this, we rely on the existence of a non-Ramsey collection, which is treated here as a nonconstructive object.

\begin{flushleft}
\emph{Journal of Economic Literature} Classification Numbers: \texttt{D60,} \texttt{D70,} \texttt{D90.}
\end{flushleft}
\emph{Keywords and Phrases}: \texttt{Construction,}  \texttt{Hammond Equity,} \texttt{Lexicographic order,} \texttt{Non-Ramsey Set,} \texttt{Pigou-Dalton transfer principle,} \texttt{Representation,} \texttt{Social Welfare Orders,} \texttt{Strong Equity.}
\end{abstract}

\newpage
\section{Introduction}

\emph{Intergenerational equity} deals  with the issue of how to treat the well-being of future generations relative to the welfare of those living at present.
In case the future extends indefinitely, we are required to evaluate infinite utility streams consistently with social preferences that respect suitable equity and efficiency requirements.

We focus on social welfare orders (SWOs), complete and transitive binary relations on the set \(X\) of infinite utility streams, where \(X\) takes the form \(X=Y^{\N}\), with \(Y\) being a non-empty subset of real numbers \(\R\), and \(\N\) being the set of natural numbers.
A social welfare function (SWF) is a real-valued map on \(X\).
If an SWO can be represented by an SWF, then the social ranking of any two utility streams can be recovered from the corresponding real values assigned by the SWF.
The issue of representation is, however, conceptually different from the issue of construction of the SWOs.
The latter concept involves the restriction that the use of the \emph{Axiom of Choice} (AC), or a similar non-constructive device, is not permitted.
In contrast, a real-valued representation may itself be obtained by using the Axiom of Choice or a similar device.

The equity principles that we study in this paper, involve an alteration in the distribution of utilities and are called consequentialist equity concepts.
Hammond Equity and the Pigou-Dalton transfer principle are two important consequentialist equity concepts.
Strong Equity is a notable strengthening of Hammond Equity (see \citet{aspremont1977}).
The issues of representation and construction of SWOs on infinite utility streams satisfying some notion of consequentialist equity is the central goal here.

Strong Equity involves comparisons between two utility streams, \(x\) and \(y\), in which all generations except two have the same utility levels in both utility streams.
Regarding the two remaining generations, say \(i\) and \(j\), one of the generations, say \(i\), is better off in utility stream \(x\), and the other generation, \(j\), is better off in utility stream \(y\).
The axiom states that if, in both utility streams, it is generation \(i\) which is worse off than generation \(j\), then generation \(i\) should be allowed, on behalf of society, to choose between \(x\) and \(y\).
That is, \(x\) is socially preferred to \(y\), since generation \(i\) is better off in \(x\) than in \(y\).
Hammond Equity yields a weak ranking in the same setting.
The Pigou-Dalton transfer principle, shows preference for a non-leaky and non-rank-switching (precise definition in Section 2) redistribution of a good from a \emph{rich} person to a \emph{poor} person, with no one else being affected.
In this transfer, the gain to the poor person exactly equals the loss suffered by the rich person.
Further the transfer is non-rank-switching in the sense that the poor person does not end up having more than the rich person.

It is well known that SWOs satisfying consequentialist equity axioms may exist, a result proven relying on non-constructive methods.
For Hammond Equity, \citet{bossert2007} have shown that there exist SWOs on infinite utility streams satisfying Hammond Equity and Strong Pareto.
It is also known that an SWO satisfying PD exists; see \citet[Theorem 1]{bossert2007}.
However, these existence results use the variant of Szpilrajn's Lemma given in \citet{arrow1951}, which is a non-constructive device.
It is pertinent to note that if the SWO is shown to exist using some form of AC then it remains an open question whether the SWO admits an explicit construction.
An affirmative answer would at least enable the social planner to rank pairs of infinite utility streams and would be useful to policy makers.

\citet{fleurbaey2003} explored the necessity of use of such non-constructive techniques and conjectured that it may not be possible to explicitly describe the SWOs that have been shown to exist using AC.
The conjecture, if found to be true, could have important consequences since the SWO would then be of limited practical use in policy making.
\citet{lauwers2010} and \citet{zame2007} using alternative techniques proved the necessity of reliance on some non-constructive device to establish the existence result, thereby confirming the conjecture.
\citet{lauwers2010} relied upon the existence of non-Ramsey sets to establish the non-constructive nature of the SWOs. 
We rely on the existence of non-Ramsey sets to establish non-constructive nature of SWOs.

Axiom of Choice has also been relied upon to prove the existence of representable SWOs.
In two recent contributions, \citet[Proposition 5]{alcantud2010a} and \citet[Proposition 3]{sakamoto2012} have proved the existence of a representable SWO satisfying PD, when the domain set is $Y = [0, 1]$.
The techniques used by them are similar to that introduced by \citet[Proposition 1]{basu2007}, who show the existence of a representable SWO satisfying the anonymity and an efficiency principle known as \emph{weak dominance}. 
Since, these existence results use AC it could turn out to be the case that the representations are non-constructive.


In this paper, we examine the constructive features of SWOs and representable SWOs that have been shown to exist in the literature.
Our results are summarized as follows.

Our first contribution is a positive construction result.
We show that if \(Y\) is well-ordered by the usual order, then there exists an explicitly defined SWO on \(X=Y^{\N}\) satisfying Strong Pareto and Strong Equity.
Since Strong Equity implies Hammond Equity, and since the present strict formulation of PD is covered by Strong Equity together with the additional conservation requirement, the constructed SWO also satisfies Hammond Equity and PD (see Remark \ref{R1}).

The proof constructs an explicitly defined social welfare ordering by means of a lexicographic comparison. 
Crucially, the lexicographic ordering is not imposed directly on generations. 
Rather, utility streams are mapped into binary threshold-utility codes, and the resulting codes are ordered lexicographically. 
This distinction is essential: it allows the ordering to preserve the desired equity-sensitive comparisons while retaining completeness, transitivity, and Strong Pareto efficiency.
Thus, on well-ordered domains, the consequentialist equity axioms can be jointly satisfied with Strong Pareto by an explicitly specified social welfare ordering.

The second contribution concerns nonconstructive implications.
We show that if \(Y\) contains a subset of order type \(\omega^\ast\),%
\footnote{The set of negative integers ordered by the usual order is an easy example.} and if there exists an SWO on \(X\) satisfying Strong Equity and Monotonicity, then there exists a non-Ramsey collection of infinite subsets of \(\N\).
We also prove an analogous result for PD together with Monotonicity under a constant-gap descending-block condition on the domain.
Thus, for sufficiently rich descending domains, the existence of an SWO satisfying these axioms entails the existence of a nonconstructive object.

The third contribution concerns representation and construction for Strong Equity alone.
We first show that if \(Y\) is countable, then there exists an SWF \(W:X\to\R\) satisfying Strong Equity.
The proof uses an extension argument on eventual-equivalence classes and is therefore non-explicit in general.
We then show that no real-valued SWF satisfying Strong Equity exists whenever \(Y\) is uncountable.
Together with the positive result for countable domains, this yields a complete characterization of the domains admitting a real-valued SWF satisfying Strong Equity, even without adding any efficiency axiom.
For finite domains, we also examine explicit construction and continuity of SWFs satisfying Strong Equity.

Finally, we provide explicit SWO constructions for Strong Equity alone under order-theoretic restrictions on the domain.
Let \(L(Y)\) be the set of utility levels that can occur as the pre-transfer poor generation in a Strong Equity comparison, and let \(U(Y)\) be the set of utility levels that can occur as the pre-transfer rich generation.
If \(L(Y)\) is well-ordered by the usual order, or if \(U(Y)\) is reverse well-ordered, then there exists an explicitly defined SWO on \(X\) satisfying Strong Equity.
On the negative side, if \(Y\) contains an integer-like subset,%
\footnote{We call them as subset of order type \(\mu\) following nomenclature introduced in \citet{dubey2011}.} 
then the existence of an SWO satisfying Strong Equity implies the existence of a non-Ramsey collection.
This identifies a clear class of domains on which Strong Equity rankings are nonconstructive in the sense used in this paper.

A further contribution is to clarify the scope of earlier correspondence claims relating real-valued representation, explicit SWO construction.
The results obtained here show that the domain on which an SWF exists need not coincide with the domain on which an explicitly constructible SWO exists.
Thus the relationship between representation and construction is axiom-specific and depends on the order-theoretic and, in the case of PD, additive structure of the utility domain.

The rest of the paper is organized as follows.
Section \ref{sec2} introduces the notation, the equity and efficiency axioms, the relevant order types, and Ramsey and non-Ramsey collections.
Section \ref{sec3} gives explicit constructions for well-ordered domains and proves non-constructiveness results for Strong Equity and PD under Monotonicity.
Section \ref{sec4} studies representation and construction for Strong Equity alone, including the characterization of real-valued representation by countability of the domain and a non-Ramsey obstruction on integer-like domains.
Section \ref{sec5} revisits the correspondence claims in \citet{dubey2014}, \citet{dubey2014b}, \citet{dubey2016}, and \citet{dubey2016a} in light of the present results.

\section{Preliminaries}\label{sec2}

\subsection{Notation}

Let $\R$, $\Z$, and $\N$ denote the sets of real numbers, integers, and natural numbers, respectively; throughout, $\N=\{1,2,\cdots\}$.
For all $y, z\,\in \R^{\N}$, we write $y\geq z$ if $y_{n}\geq z_{n}$, for all $n\in \N$; we write $y>z$ if $y\geq z$ and $y\neq z$; and we write $y\gg z$ if $y_{n}>z_{n}$ for all $n\in \N$.

\subsection{Definitions}

Let $Y$, a non-empty subset of $\R$, be the set of all possible utilities that any generation can achieve. 
Then $X\equiv Y^{\N}$ is the set of all possible utility streams. 
If $x \;\in X$, then $x = (x_{1}, x_{2},\cdots)$, where, for all $n\in \N$, $x_{n}\in Y$.
We consider binary relations on $X$, denoted by $\succsim$, with symmetric and asymmetric parts denoted by $\sim $ and $\succ$ respectively, defined in the usual way. 
A \emph{social welfare order} (SWO) is a complete and transitive binary relation.
A \emph{social welfare function} (SWF) is a mapping $W:X\rightarrow \mathbb{R}$. 
Given a SWO $\succsim$ on $X$, we say that $\succsim$ can be \emph{represented} by a real-valued function if there is a mapping $W: X\rightarrow \mathbb{R}$ such that for all $x,y\in X$, we have $x\succsim y$ if and only if $W(x)\geq W(y)$.
We say that an SWF $W$ satisfies an axiom stated for SWOs if the induced relation $x\succeq_W y$ defined by $W(x)\geq W(y)$ satisfies that axiom.

\subsubsection{Equity and Efficiency axioms}

The following consequentialist equity and efficiency axioms on social welfare orders are used in this paper. 

\begin{definition}[Monotonicity (M)]
\label{D1}
\emph{An SWO $\succeq$ on $X$ satisfies monotonicity if, for all $x,y\in X$, $x\geq y$ then $x\succeq y$.}
\end{definition}

\begin{definition}[Strong Pareto (SP)]
\label{D2}
An SWO $\succeq$ on $X$ satisfies Strong Pareto if, for all $x,y\in X$, $x>y$ then $x\succ y$.
\end{definition}

\begin{definition}[Weak Pareto (WP)]
\label{D3}
An SWO $\succeq$ on $X$ satisfies Weak Pareto if, for all $x,y\in X$, $x\gg y$ then $ x\succ y$.
\end{definition}

\begin{definition}[Strong Equity (SE)]
\label{D4}
\emph{An SWO $\succeq$ on $X$ satisfies Strong Equity if, for all $x,y\in X$, whenever there exist $i,j\in\N$ such that $y_j>x_j>x_i>y_i$ and $x_k=y_k$ for all $k\in\N\setminus\{i,j\}$, then $x\succ y$.}
\end{definition}

\begin{definition}[Hammond Equity (HE)]
\label{D5}
\emph{An SWO $\succeq$ on $X$ satisfies Hammond Equity if, for all $x,y\in X$, whenever there exist $i,j\in\N$ such that $y_j>x_j>x_i>y_i$ and $x_k=y_k$ for all $k\in\N\setminus\{i,j\}$, then $x\succeq y$.}%
\footnote{
It is useful to note that there exists a real-valued SWF satisfying HE.
For instance, for all $x\in X$, let $W(x)= \liminf{x_t}$. 
Then $W$ is monotone and satisfies HE. 
In addition, it also satisfies a very weak efficiency condition, namely Uniform Improvement Pareto (UIP): for all $x, y\in X$, if $x_t>y_t+\varepsilon$ for some $\varepsilon>0$ and all $t\in \N$, then $W(x)>W(y)$.
Thus there exists a monotone, explicitly defined SWF satisfying HE and UIP.}\label{FN1} 

\end{definition}

\begin{definition}[Pigou--Dalton transfer principle (PD)]
\label{D6}
\emph{An SWO $\succeq$ on $X$ satisfies the Pigou-Dalton transfer principle if, for all $x,y\in X$, whenever there exist $i,j\in\N$ such that $y_j>x_j>x_i>y_i$, $y_j+y_i=x_j+x_i$, and $x_k=y_k$ for all $k\in\N\setminus\{i,j\}$ then $x\succ y$.}
\end{definition}

\subsection{Domain Types}

In this subsection, we recall standard terminology concerning order types.
A set $S$ is \emph{strictly ordered} by a binary relation $<$ if $<$ is connected, transitive, and irreflexive: for all $s, s^{\prime}\in S$ with $s \neq s^{\prime}$, either $s<s^{\prime}$ or $s^{\prime}<s$ holds; if $s<s^{\prime}$ and $s^{\prime}<s^{\prime\prime}$, then $s<s^{\prime\prime}$; and $s<s$ holds for no $s\in S$.
We denote such a strictly ordered set by $S(<)$.

Two strictly ordered sets $S(<)$ and $S^{\prime}(\prec)$ are \emph{similar} if there exists a one-to-one map $f:S\to S^{\prime}$ onto $S^{\prime}$ such that, for all $s_1,s_2\in S$, $s_1<s_2$ if and only if $f(s_1)\prec f(s_2)$.

Let $Y$ be a non-empty subset of $\R$.
Let $\N^{-}={-1,-2,-3,\cdots}$, ordered by the usual order.
We use the following order types:
\begin{itemize}
\item $Y(<)$ is of order type $\omega$ if $Y(<)$ is similar to $\N(<)$.

\item $Y(<)$ is of order type $\omega^{\ast}$ if $Y(<)$ is similar to $\N^{-}(<)$.

\item $Y(<)$ is of order type $\sigma$ if $Y(<)$ is similar to $\I(<)$.

\item $Y(<)$ is of order type $\mu$ if $Y$ contains a non-empty subset $Y^{\prime}$ such that $Y^{\prime}(<)$ is of order type $\sigma$.
\end{itemize}
Thus, a set of order type $\mu$ contains an integer-like subset.

\subsection{Ramsey and Non-Ramsey Collections of Sets}

Let $T$ be an infinite subset of $\N$.
We denote by $\calS(T)$ the collection of all infinite subsets of $T$.
We denote $\calS(\N)$ by $\calS$.

\begin{definition}[Non-Ramsey collection]
\label{D7}
\emph{A collection $\OmegaC\subseteq\calS$ is Ramsey if there exists $T\in\calS$ such that either $\calS(T)\subseteq\OmegaC$ or $\calS(T)\subseteq\calS\setminus\OmegaC$.
The collection $\OmegaC$ is non-Ramsey if, for every $T\in\calS$, the collection $\calS(T)$ intersects both $\OmegaC$ and $\calS\setminus\OmegaC$.}
\end{definition}

\section{Monotone SWO satisfying Consequentialist Equity} \label{sec3}
In this section we examine the constructive status of social welfare orders satisfying SE, HE, or PD together with efficiency requirements. 
The results are mixed. 
For some domains ($Y\subseteq\mathbb R$), such SWOs admit explicit construction; for others, their existence entails the existence of a non-Ramsey collection. 
The order type of subsets of $Y$ plays the critical role.

Since our results cover both constructive and nonconstructive conclusions, the choice of efficiency axiom requires some explanation. 
In the constructive theorem, we impose the strongest efficiency requirement, Strong Pareto. 
This is not an additional burden on the construction: the threshold utility-lexicographic order constructed below satisfies Strong Pareto, and hence the positive result has maximal strength with respect to the efficiency hierarchy. 
By contrast, in the nonconstructive results we use the weakest efficiency requirement that still drives the obstruction. 
For SE and PD, the efficiency condition turns out to be monotonicity. 
For HE, monotonicity alone is too weak, since universal indifference satisfies HE  and monotonicity; hence the negative result in case of HE requires a stronger Pareto-type condition.

\subsection{Explicit Description of SWO satisfying Strong Equity and Strong Pareto}

We show that if \(Y\) is well-ordered by the usual order, then there exists an explicitly described social welfare order on \(X=Y^{\mathbb N}\) satisfying SP, SE, HE and PD.
The construction is lexicographic, not over generations directly, but over binary threshold utility codes.

\begin{theorem}
\label{T1}
\emph{Let \(Y\subseteq\mathbb R\) be non-empty and well-ordered by the usual order \(<\), and let \(X\equiv Y^{\mathbb N}\). 
Then there exists an explicitly defined SWO on \(X\) satisfying Strong Pareto and Strong Equity.}
\end{theorem}

\begin{proof}
Let $\bar{Y}\equiv \{t\in Y: \text{there exists } s\in Y \text{ with } t<s\}$.
Observe that since $Y$ is well-ordered by $<$, the set $\bar{Y}$ is also well-ordered by $<$.
For $x\in X$, $t\in \bar{Y}$, and $n\in\N$, define
\[
B_{t,n}(x)=
\begin{cases}
1, & x_n\leq t,\\
0, & x_n>t.
\end{cases}
\]
We define a lexicographic order on $\bar{Y}\times\N$ by comparing $t$ (call it threshold utility) first and $n$ (generation) second.
In other words, $(t,n)<(s,m)$ if either $t<s$, or $t=s$ and $n<m$.
Thus the index set $\bar{Y}\times\N$ is well-ordered in this lexicographic order.
For $x,y\in X$, define $x\succeq^{\ast} y$ as follows.
If $B_{t,n}(x)=B_{t,n}(y)$ for every $(t,n)\in \bar{Y}\times\N$, set $x\sim^{\ast} y$.
If the two codes differ, let $(t^{\ast},n^{\ast})$ be the first pair at which they differ.
The strict ranking is defined as follows: 
\[
x\succ^{\ast} y, \;\text{if and only if}\; B_{t^{\ast},n^{\ast}}(x)=0 \; \text{and} \; B_{t^{\ast},n^{\ast}}(y)=1.
\]
Thus the lexicographic order on $\{B_{t,n}(x)\}_{(t,n)\in \bar{Y}\times\N}$, ranks $0$ above $1$.

\noindent\textbf{Completeness.}
If all threshold utility codes of $x$ and $y$ agree, then $x=y$.
If $x\neq y$, then there is a minimal $n\in \N$ such that $x_n\neq y_n$.
If $x_n>y_n$, then $y_n\in \bar{Y}$ and
\[
B_{y_n,n}(x)=0<1=B_{y_n,n}(y).
\]
The case $y_n>x_n$ is symmetric.

\noindent\textbf{Transitivity.}
It follows from transitivity of the lexicographic binary relation.

Hence, $\succeq^{\ast}$ is an SWO on $X$.

\noindent We show next that $\geq_{\lex}$ satisfies Strong Pareto and Strong Equity axioms.

\noindent\textbf{Strong Pareto.}
Let $x>y$.
For every $t\in \bar{Y}$ and $n\in\N$, $x_n\geq y_n$ implies
\[
B_{t,n}(x)\leq B_{t,n}(y).
\]
Since $x\neq y$, choose $m\in\N$ with $x_m>y_m$.
Since \(B_{t,n}(x)\leq B_{t,n}(y)\) for every \((t,n)\), no code difference could prefer \(y\). 
Since \(B_{y_m,m}(x)=0\) and \(B_{y_m,m}(y)=1\), there is at least one code difference that prefers \(x\). 
Hence the first code difference favors \(x\), and therefore \(x\succ^\ast y\).

\noindent\textbf{Strong Equity.}
Let $x,y\in X$ and let $i,j\in\N$ satisfy $y_j>x_j>x_i>y_i$ and $x_k=y_k$ for all $k\in\N\setminus\{i,j\}$.
For every threshold utility $s<y_i$, the generations $i$ and $j$'s utility exceed $s$ in both streams, so $B_{s, i} (x) = B_{s, i}(y)$. 
At threshold utility $y_i$, 
\[
B_{y_i,i}(x)=0, \; B_{y_i,i}(y)=1, \;\text{and}\;  B_{y_i,j}(x)=B_{y_i,j}(y)=0.
\]
Since, $x_k = y_k$ for all $k\in\N\setminus\{i,j\}$, the first code difference occurs at $y_i, i$ and shows strict preference for $x$.
Hence $x\succ^{\ast} y$.
\end{proof}

\begin{remark}
\label{R1}
\emph{\begin{itemize}
\item Since Strong Equity implies Hammond Equity, the SWO constructed in Theorem \ref{T1} also satisfies Hammond Equity.
Indeed, Hammond Equity requires only weak preference in the same two-generation comparisons in which Strong Equity requires strict preference.
\item Under the present strict definition of the Pigou-Dalton transfer principle, the SWO constructed in Theorem \ref{T1} also satisfies PD.
To see this, let \(i,j\in\mathbb N\) satisfy
\[
y_j>x_j>x_i>y_i,\; y_j+y_i=x_j+x_i,
\]
with all other generations unchanged.
The inequalities \(y_j>x_j>x_i>y_i\) already form a Strong Equity comparison.
Equivalently, in the construction used in Theorem \ref{T1}, at the threshold \(t=y_i\), generation \(i\) changes from weakly below \(t\) in \(y\) to above \(t\) in \(x\), while generation \(j\) remains above \(t\) in both streams.
Thus the first relevant threshold utility-code difference shows strict preference for \(x\), so \(x\succ^\ast y\).
\item The order constructed above admits an explicit description, but it is lexicographic and therefore should not be confused with a real-valued social welfare function representation.
Existing impossibility results show that, on many nontrivial domains, equity requirements together with Pareto requirements rule out real-valued representation.
The lexicographic order above should therefore be read as an explicit SWO construction, not as a real-valued SWF representation.
\end{itemize}}
\end{remark}

\begin{remark}\label{R2}
\emph{\begin{itemize}
\item 
\citet[Proposition 1]{dubey2016} shows that, when \(Y\) is the set of negative integers, the existence of an SWO on \(Y^{\N}\) satisfying HE and WP entails the existence of a non-Ramsey set. 
The same conclusion holds whenever \(Y\) contains a subset \(Z\) of order type \(\omega^{\ast}\). 
Indeed, any SWO on \(Y^{\N}\) restricts to an SWO on \(Z^{\N}\), and \(Z\) can be relabelled in an order-preserving way as the negative integers. Since HE and WP are ordinal conditions, such a relabelling preserves both axioms. Hence Dubey's negative-integer argument applies to every \(\omega^{\ast}\)-subdomain.
\item 
These observations give a fairly complete analysis of the Pareto variants considered here for SWOs satisfying HE. 
No SWF satisfying HE and SP exists on any domain \(Y\) containing more than three distinct elements. 
SWOs satisfying HE and SP admit an explicit description when \(Y\) is well-ordered. 
On the other hand, if \(Y\) contains a subset of order type \(\omega^{\ast}\), then any SWO satisfying HE and WP is non-constructive. 
Weakening the Pareto requirement further to Uniform Improvement Pareto allows one to obtain an SWF satisfying HE (see footnote \ref{FN1}).
\end{itemize}}
\end{remark}

\subsection{SWO implies existence of Non-Ramsey set}
We discuss the non-constructive nature of SWOs satisfying  SE and PD in the following theorem.

\begin{theorem}
\label{T2}
\emph{Let \(Y\subseteq\R\) be non-empty and let \(X\equiv Y^\N\).
The following statements hold.
\begin{enumerate}[label=\textup{(\alph*)}]
\item 
If \(Y\) contains a subset of order type \(\omega^{\ast}\), and if there exists an SWO on \(X\) satisfying Strong Equity and monotonicity, then there exists a non-Ramsey collection \(\OmegaC\subseteq\calS\).
\item 
Suppose that there exist sequences \((\alpha_m)_{m\in\N}\) and \((\beta_m)_{m\in\N}\) in \(Y\) and a number \(\delta>0\) such that
\[
\alpha_m<\beta_m,\;  \beta_m-\alpha_m=\delta
\;\forall m\in\N,\; \text{and for}\; m>l, 
\alpha_m<\beta_m<\alpha_\ell<\beta_\ell.
\]
If there exists an SWO on \(X\) satisfying the Pigou--Dalton transfer principle and monotonicity, then there exists a non-Ramsey collection \(\OmegaC\subseteq\calS\).
\end{enumerate}}
\end{theorem}

The proof of the theorem will be based on a series of lemmas.
For part \textup{(a)}, if $Y$ contains a subset of order type $\omega^{\ast}$, choose a strictly decreasing sequence $c_1>c_2>\cdots$ in $Y$ and set $\beta_m=c_{2m-1}$ and $\alpha_m=c_{2m}$. Then $\alpha_m<\beta_m$ for every $m\in\N$, and for $m>\ell$, $\alpha_m<\beta_m<\alpha_\ell<\beta_\ell$. For part \textup{(b)}, use the sequences assumed in the statement of the theorem.
Throughout the following lemmas, let $\{\alpha_m,\beta_m\}_{m\in\N}$ be a sequence of pairs in $Y$ such that $\alpha_m<\beta_m \; \text{for every } m\in\N$ and, for  $m>\ell$, $\alpha_m<\beta_m<\alpha_\ell<\beta_\ell$.
We partition $\N$ into consecutive finite blocks $C_m$ with $|C_m|=2^{m-1}$, where $|\cdot|$ denotes the cardinality of the set.
For $A\subseteq\N$, define the sequence $\Phi(A)\in Y^\N$ by %
\footnote{For example, we  take $C_m=\{2^{m-1},2^{m-1}+1, \cdots, 2^m-1\}$, and $\alpha_m=-2m-1$, $\beta_m=-2m$.
Then \(\beta_m-\alpha_m=1\) for every \(m\), and
\(\alpha_m<\beta_m<\alpha_\ell<\beta_\ell\) whenever \(m>\ell\).
Thus, for \(A=\{1,3,5,\cdots\}\), we get $\Phi(A)=(-2,-5,-5,-6,-6,-6,-6,-9,\cdots,-9,-10,\cdots)$, where each block \(C_m\) is filled with \(\beta_m\) if \(m\in A\), and with \(\alpha_m\) otherwise.}
\[
\Phi(A)_n=
\begin{cases}
\beta_m, & n\in C_m \text{ and } m\in A,\\
\alpha_m, & n\in C_m \text{ and } m\notin A.
\end{cases}
\]
For an infinite set $S=\{s_1<s_2<s_3<\cdots\}\in\calS$, define
\[
S_1=[s_2,s_3)\cup [s_4,s_5)\cup [s_6,s_7)\cup\cdots
\]
and
\[
S_2=[s_1,s_2)\cup [s_3,s_4)\cup [s_5,s_6)\cup\cdots,
\]
where $[r,s)=\{n\in\N:r\leq n<s\}$.

\begin{lemma}
\label{L1}
Let $A,B\subseteq\N$ and $A\subseteq B$.
Then $\Phi(B)\geq\Phi(A)$.
If $\succeq$ is monotone then,  $\Phi(B)\succeq\Phi(A)$.
\end{lemma}

\begin{proof}
If $m\notin B$, then $m\notin A$, so both streams assign $\alpha_m$ to every coordinate in $C_m$.
If $m\in A$, then $m\in B$, so both streams assign $\beta_m$ to every coordinate in $C_m$.
If $m\in B\setminus A$, then $\Phi(B)$ assigns $\beta_m$ and $\Phi(A)$ assigns $\alpha_m$ to every coordinate in $C_m$.
Since $\beta_m>\alpha_m$, we have $\Phi(B)\geq\Phi(A)$. Hence, if $\succeq$ is monotone, then $\Phi(B)\succeq\Phi(A)$.
\end{proof}

\begin{lemma}
\label{L2}
Let  SWO $\succeq$ on $Y^\N$ satisfy SE and M.
Let $A,B\subseteq\N$ be such that $A\setminus B$ is finite and non-empty.
If there exists $m\in B\setminus A$ such that $m>\max(A\setminus B)$, then $x = \Phi(B) \succ \Phi(A) = y$.
\end{lemma}

\begin{proof}
Let
\[
P=\bigcup_{\ell\in A\setminus B}C_\ell .
\]
Choose \(m\in B\setminus A\) such that \(m>\max(A\setminus B)\). Since the blocks \(C_m\) are expanding, we have
\[
|P|\leq \sum_{\ell<m}|C_\ell|<|C_m|.
\]
Hence we may choose an injective map
\[
\rho: P\to C_m.
\]
Starting from \(y=\Phi(A)\), we modify the coordinates in \(P\) as follows. 
If \(p\in C_\ell\subseteq P\), change coordinate \(p\) from \(\beta_\ell\) to \(\alpha_\ell\), and simultaneously change coordinate \(\rho(p)\in C_m\) from \(\alpha_m\) to \(\beta_m\). 
Since \(m>\ell\), we have
\[
\beta_\ell>\alpha_\ell>\beta_m>\alpha_m.
\]
Thus each two-coordinates change is strictly preferred by applying SE axiom and there are only finitely many such pairs.
We denote the utility stream obtained by performing these finitely many changes by \(z\). 
By repeated applications of SE and transitivity, we get $z\succ \Phi(A)=y$.
Next we consider  \(z\) and \(\Phi(B)=x\). 
On every coordinate in \(P\), the utility stream \(z\) has the same value as in \(\Phi(B)\), since both assign the lower value \(\alpha_\ell\) on each block \(C_\ell\) with \(\ell\in A\setminus B\). 
On every coordinate in \(\rho(P)\subseteq C_m\), the stream \(z\) also agrees with \(\Phi(B)=x\), since both assign \(\beta_m\). 
On the remaining coordinates of \(C_m\), however, \(z\) still assigns \(\alpha_m\), while \(\Phi(B)\) assigns \(\beta_m\). 
Similarly, on any other block \(C_j\) with \(j\in B\setminus A\), \(\Phi(B)\) assigns \(\beta_j\), while \(z\) assigns \(\alpha_j\). 
On all other blocks the two streams agree. 
Therefore, $x=\Phi(B)\geq z$, and by monotonicity, $x=\Phi(B)\succeq z$.
Combining this with \(z\succ \Phi(A)=y\), we obtain
\[
x=\Phi(B)\succ \Phi(A)=y.
\]
\end{proof}

\begin{lemma}
\label{L3}
\emph{Suppose there exists $\delta>0$ such that $\beta_m-\alpha_m=\delta$ for every $m\in\N$.
Let $\succeq$ be an SWO on $Y^\N$ satisfying the Pigou-Dalton transfer principle and monotonicity.
Let $A,B\subseteq\N$ with $A\setminus B$ is finite and non-empty and there exists $m\in B\setminus A$ such that $m>\max(A\setminus B)$.
Then $x=\Phi(B)\succ\Phi(A)=y$.}
\end{lemma}

\begin{proof}
Let $P$ be as in Lemma \ref{L2}.
The same cardinality calculation gives an injection $\rho:P\to C_m$.
Start from $\Phi(A)$.
For each $p\in P$, let $p\in C_\ell$ and let $q=\rho(p)\in C_m$ be its assigned coordinate.
Change coordinate $p$ from $\beta_\ell$ to $\alpha_\ell$ and coordinate $q$ from $\alpha_m$ to $\beta_m$.
At that two-coordinate step, 
\[
y_p=\beta_\ell>x_p=\alpha_\ell>x_q=\beta_m>y_q=\alpha_m.
\]
The equal-gap condition gives
\[
\beta_\ell-\alpha_\ell=\beta_m-\alpha_m \iff 
y_p+y_q=x_p+x_q.
\]
Because $m>\ell$, we also have $x_p=\alpha_\ell>x_q=\beta_m$.
All other coordinates are unchanged in that paired step.
Each step is a PD improvement from $y$ to $x$: utility is transferred from the better-off coordinate to the worse-off coordinate, total utility over the two coordinates is unchanged, and the resulting pair is less  unequal. 
We denote the utility sequence obtained from performing these finitely many PD transfers by \(z\). 
By finitely many PD transfers and transitivity, we get $z\succ \Phi(A)$.

It remains to compare \(z\) with \(\Phi(B) =x\). 
On every coordinate in \(P\), the stream \(z\) agrees with \(\Phi(B)\), since both assign \(\alpha_\ell\) on each block \(C_\ell\) with \(\ell\in A\setminus B\). 
On every coordinate in \(\rho(P)\subseteq C_m\), the stream \(z\) also agrees with \(\Phi(B)\), since both assign \(\beta_m\). 
On the remaining coordinates of \(C_m\), \(z\) still assigns \(\alpha_m\), while \(\Phi(B)\) assigns \(\beta_m\). 
On any other block \(C_j\) with \(j\in B\setminus A\), the stream \(z\) assigns \(\alpha_j\), while \(\Phi(B)\) assigns \(\beta_j\). 
On all remaining blocks the two streams agree. 
Therefore, $\Phi(B)\geq z$.
By monotonicity, $\Phi(B)\succeq z$.
Consequently,
\[
x=\Phi(B)\succeq z\succ \Phi(A)=y \implies  x=\Phi(B)\succ \Phi(A)=y.
\]
\end{proof}

\begin{proof} [\textbf{Proof of Theorem \ref{T2}}]
We apply the common block construction above. For part \textup{(a)}, Lemma~\ref{L2} supplies the required strict comparisons; for part \textup{(b)}, Lemma~\ref{L3} supplies them.
Let \(\succeq\) be the assumed SWO. For \(S\in\calS\), define
\[
S\in\OmegaC
\;\Longleftrightarrow\;
\Phi(S_1)\succ\Phi(S_2).
\]
We show that \(\OmegaC\) is non-Ramsey. 
Thus, for each \(S\in\calS\), we must show that \(\calS(S)\) intersects both \(\OmegaC\) and its complement \(\calS\setminus\OmegaC\).
It is enough to prove that for every \(S\in\calS\), there exists \(T\in\calS(S)\) such that
\[
T\in\OmegaC
\;\Longleftrightarrow\;
S\notin\OmegaC .
\]
Fix $S=\{s_1, s_2, s_3,\cdots\}\in\calS$ with $s_1< s_2 <s_3<\cdots$, and let $I_k=[s_k, s_{k+1})$.
We consider the following two cases.
\begin{enumerate}
\item
Suppose $S\in\OmegaC$.
Then $\Phi(S_1)\succ\Phi(S_2)$.
Drop $s_1$ to obtain $T=\{s_2,s_3,s_4,\cdots\}$.
Then $T\in\calS(S)$ and $T_1=S_2\setminus I_1$ and $T_2=S_1$.
Since $T_1\subseteq S_2$, Lemma \ref{L1} gives $\Phi(S_2)\succeq\Phi(T_1)$.
Since $S\in\OmegaC$, we have
\[
\Phi(T_2)=\Phi(S_1)\succ\Phi(S_2)\succeq\Phi(T_1) \implies 
\Phi(T_2)\succ\Phi(T_1).
\]
Therefore $T\notin\OmegaC$.
Thus $\calS(S)$ intersects $\calS\setminus\OmegaC$.

\item
Suppose $S\notin\OmegaC$. By completeness, $\Phi(S_2)\succeq\Phi(S_1)$.
Drop $s_1$, $s_{4n}$, $s_{4n+1}$ for all $n\in \N$ to obtain 
\[
T=\{s_2,s_3,s_6,s_7,s_{10},s_{11},\cdots\}.
\]
Then $T\in\calS(S)$.
For this $T$, we have
\[
T_1=(I_3\cup I_4\cup I_5)\cup(I_7\cup I_8\cup I_9)\cup\cdots\;\text{and}\;  T_2=I_2\cup I_6\cup I_{10}\cup\cdots.
\]
Thus, $S_2\setminus T_1=I_1$ is finite and non-empty, while $
T_1\setminus S_2=I_4\cup I_8\cup I_{12}\cup\cdots$ is infinite and unbounded.
In particular, there exists $m\in T_1\setminus S_2$ such that $m>\max(S_2\setminus T_1)$.
For part \textup{(a)}, Lemma \ref{L2} gives $\Phi(T_1)\succ\Phi(S_2)$.
For part \textup{(b)}, Lemma \ref{L3} gives the same strict preference.
Since $T_2\subseteq S_1$, Lemma \ref{L1} gives $\Phi(S_1)\succeq\Phi(T_2)$.
Therefore, 
\[
\Phi(T_1)\succ\Phi(S_2)\succeq\Phi(S_1)\succeq\Phi(T_2) \implies \Phi(T_1)\succ\Phi(T_2).
\]
So $T\in\OmegaC$ and  $\calS(S)$ intersects $\OmegaC$.

\end{enumerate}
\end{proof}

\section{SWO satisfying Strong Equity: Representation and Construction}\label{sec4}

\subsection{Representation}
In this section, we investigate social welfare orders satisfying Strong Equity. 
Unlike Hammond Equity, Strong Equity itself generates strict comparisons between infinite utility streams. 
It is therefore natural to study the structure of the strict relation generated by Strong Equity, both from the viewpoint of real-valued representation by a social welfare function and from the viewpoint of explicit construction of a social welfare order. 
The results depend crucially on the structure of the utility domain \(Y\).

The following proposition gives a positive representation result: when \(Y\) is countable, there exists a social welfare function satisfying Strong Equity.

\begin{proposition}
\label{P1}
Let \(Y\subseteq\R\) be non-empty and countable, and $X=Y^{\N}$. 
Then there exists a social welfare function $W:X\to\R$ satisfying Strong Equity. 
\end{proposition}

\begin{proof}
We partition \(X\) into equivalence classes under the relation of eventual equality: two streams \(x,y\in X\) belong to the same class if there exists \(N\in\N\) such that $x_n=y_n$ for all $n\geq N$.
Fix an eventual-equivalence class \(E\). 
For \(x,y\in E\), we say $x\triangleright_E y$ if \(x\) can be obtained from \(y\) by finitely many one-step Strong Equity improvements. 
Equivalently, each step replaces two coordinates of the form $(a,d)\mapsto(b,c)$, with $a<b<c<d$, while all the remaining utility levels stay unchanged.

We first show that \(\triangleright_E\) has no cycles. 
Suppose, to the contrary, that there exists a cycle. 
Expanding each comparison \(x\triangleright_E y\) into its finite chain of one-step Strong Equity improvements, we obtain a finite closed chain of one-step Strong Equity improvements. 
This closed chain involves only finitely many utility streams, with changes in finitely many coordinates that alter only finitely many utility values.
Let the finite set of utility values appearing in this closed chain be
\[
F=\{y_1,\cdots,y_K\},
\;
y_1<y_2<\cdots<y_K.
\]
We define
\[
\varphi_F(y_i)=-2^i
\;(i=1,\cdots,K).
\]
If \(y_i<y_j<y_k<y_\ell\), then \(i<j<k<\ell\), and therefore
\[
\varphi_F(y_j)+\varphi_F(y_k) = -2^j-2^k>-2^i-2^{\ell}= \varphi_F(y_i)+\varphi_F(y_\ell).
\]
Since \(j<k<\ell\), we have $j\leq \ell-2$, and $k\leq \ell-1$.
Therefore
\[
2^j+2^k
\leq 2^{\ell-2}+2^{\ell-1}
=3\cdot 2^{\ell-2}
<4\cdot 2^{\ell-2}
=2^\ell
<2^i+2^\ell .
\]
Thus the two-coordinate sum generated by \(\varphi_F\) strictly increases under every one-step Strong Equity improvement appearing in the closed chain.
Let \(J\subseteq\N\) be the finite set of coordinates used in the closed chain. 
For each utility sequence \(z\) appearing in the chain, define
\[
S_F(z)=\sum_{n\in J}\varphi_F(z_n).
\]
Observe that every one-step Strong Equity improvement in the closed chain strictly increases \(S_F\). 
Hence \(S_F\) would strictly increase around a closed cycle and return to its initial value, a contradiction. 
Therefore \(\triangleright_E\) is acyclic.

Since \(Y\) is countable, each equivalence class \(E\) is countable. 
Indeed, fixing any \(x\in E\), every element of \(E\) differs from \(x\) in only finitely many coordinates. Hence \(E\) is contained in a countable union of finite powers of the countable set \(Y\), and is therefore countable.

Let \(\triangleright_E^{+}\) denote the transitive closure of \(\triangleright_E\). 
Since \(\triangleright_E\) is acyclic, \(\triangleright_E^{+}\) is irreflexive: otherwise \(x\triangleright_E^{+}x\) for some \(x\), which would yield a finite cycle for \(\triangleright_E\). 
As \(\triangleright_E^{+}\) is transitive by definition, it is a strict partial order.
Now consider the inverse strict partial order \((\triangleright_E^{+})^{-1}\), defined by
\[
y(\triangleright_E^{+})^{-1}x
\;\Longleftrightarrow\;
x\triangleright_E^{+}y.
\]
By Szpilrajn Lemma, every strict partial order has a strict linear-order extension.%
\footnote{Equivalently, by the acyclic-extension corollary of Szpilrajn Lemma, a binary relation has a strict linear-order extension if and only if it is acyclic; see \citet[Corollary~3.1]{andrikopoulos2022}.}
Therefore there exists a strict linear order \(<_E\) on \(E\) extending \((\triangleright_E^{+})^{-1}\). Hence
\[
x\triangleright_E y
\;\Longrightarrow\;
x\triangleright_E^{+} y
\;\Longrightarrow\;
y<_E x .
\]
By the standard universality of the rational order, every countable linear order admits an order-preserving embedding into \((\mathbb Q,<)\)  (see \citet[Theorem 4.2, page 38]{jech2002}). 
Hence there exists an order-preserving injection
\[
e_E:(E,<_E)\to(\mathbb Q,<).
\]
Define
\[
W_E(x)=e_E(x)\;(x\in E).
\]
If \(x\triangleright_E y\), then \(y<_E x\), and therefore
\[
W_E(y)=e_E(y)<e_E(x)=W_E(x).
\]
Thus
\[
x\triangleright_E y
\;\Longrightarrow\;
W_E(x)>W_E(y).
\]
Perform this construction for each eventual-equivalence class $E$. Since the eventual-equivalence classes partition $X$, define $W:X\to\R$ by $W(x)=W_E(x)$ whenever $x\in E$. This is well-defined.
Now suppose that \(x\) is obtained from \(y\) by a one-step Strong Equity redistribution. 
Then \(x\) and \(y\) belong to the same equivalence class \(E\), and \(x\triangleright_E y\). 
Hence
\[
W(x)=W_E(x)>W_E(y)=W(y).
\]
Therefore \(W\) satisfies Strong Equity.
\end{proof}

Proposition \ref{P1} establishes the existence of a real-valued SWF satisfying SE when $Y$ is countable.
We now show that no real-valued SWF satisfying SE exists whenever $Y$ is uncountable.

\begin{proposition}
\label{P2}
\emph{Let $Y\subseteq\R$ be uncountable, and let $X\equiv Y^\N$.
Then there does not exist any social welfare function $W:X\to\R$ satisfying Strong Equity.}
\end{proposition}

\begin{proof}
Suppose, to the contrary, that $W:X\to\R$ satisfies Strong Equity.
Let $C$ be the set of all points $y\in Y$ such that every open interval containing $y$ contains uncountably many elements of $Y$.
Observe that $Y\setminus C$ is countable.
Indeed, each $y\in Y\setminus C$ belongs to an open interval with rational endpoints whose intersection with $Y$ is countable.
There are only countably many intervals with rational endpoints, and a countable union of countable sets is countable.
Consequently, $C$ is uncountable.

Since $C$ is uncountable, choose four points $b_0<b_1<c_1<c_0$ in $C$ such that $C$ contains a point below $b_0$ and a point above $c_0$.
Define
\[
A=\{a\in Y:a<b_0\},\; D=\{d\in Y:d>c_0\}.
\]
Both $A$ and $D$ are uncountable, by the definition of $C$.
For every $a\in A$ and $d\in D$, we have
\[
a<b_0<b_1<c_1<c_0<d.
\]

Fix $\bar y\in Y$.
For every $(a,d)\in A\times D$, define the two utility streams
\[
\begin{aligned}
x^{a,d}&=(a,d,b_0,c_0,\bar y,\bar y,\cdots),\\
y^{a,d}&=(a,d,b_1,c_1,\bar y,\bar y,\cdots).
\end{aligned}
\]
Since $b_0<b_1<c_1<c_0$, Strong Equity gives
\[
W(y^{a,d})>W(x^{a,d}).
\]
Thus $I_{a,d}=(W(x^{a,d}),W(y^{a,d}))$ is a non-empty open interval in $\R$.

We claim that if $a<a^{\prime}$ and $d>d^{\prime}$, then $I_{a,d}\cap I_{a^{\prime},d^{\prime}}=\varnothing$.
To see this, start with $y^{a,d}$ and first replace $(a,b_1)$ by $(a^{\prime},b_0)$ in coordinates 1 and 3.
This is a Strong Equity improvement because $a<a^{\prime}<b_0<b_1$.
Next replace $(c_1,d)$ by $(c_0,d^{\prime})$ in coordinates 4 and 2.
This is also a Strong Equity improvement because $c_1<c_0<d^{\prime}<d$.
The resulting stream is $x^{a^{\prime},d^{\prime}}$.
Therefore,
\[
W(y^{a,d})<W(x^{a^{\prime},d^{\prime}}),
\]
which proves the claim.

Now fix $a\in A$. 
For every $d\in D$, the interval $I_{a,d}$ is non-empty and open, and therefore contains
a rational number. 
For each $q\in\mathbb{Q}$, let
\[
D_{a,q}=\{d\in D:q\in I_{a,d}\}.
\]
Since every interval $I_{a,d}$ contains a rational number, we have
\[
D=\bigcup_{q\in\mathbb{Q}}D_{a,q}.
\]
Recall that $D$ is uncountable, while the set of rational numbers is countable. 
Therefore, it is not possible that $D_{a,q}$ is countable for every $q\in\mathbb{Q}$, since a countable union of countable
sets is countable. 
Consequently, there must exist a rational number $q_a$ such that $D_{a,q_a}$ is
uncountable. 
In other words, for our fixed $a$, the same rational number $q_a$ belongs to $I_{a,d}$ for uncountably many values of $d\in D$.

We can repeat this argument for every $a\in A$.
Thus, for each $a\in A$, choose a rational number $q_a$ such that $D_{a,q_a}$ is uncountable.
Observe that $q_a$ may depend on $a$.

For each $r\in\mathbb{Q}$, define
\[
A_r=\{a\in A:q_a=r\}.
\]
Since every $a\in A$ has been assigned a rational number $q_a$, we have
\[
A=\bigcup_{r\in\mathbb{Q}}A_r.
\]
Now $A$ is uncountable, whereas $\mathbb{Q}$ is countable.
If every set $A_r$ were countable, then $A$ would be a countable union of countable sets, and hence countable, a contradiction.
Consequently, there exists a rational number $q\in\mathbb{Q}$ such that $A_q$ is uncountable.

Let $A^{\prime}=A_q$.
Then, for every $a\in A^{\prime}$, we have $q_a=q$, and therefore $D_{a,q}$ is uncountable.
Equivalently, the same rational number $q$ belongs to $I_{a,d}$ for uncountably many values of $d\in D$, for every $a\in A^{\prime}$.

Finally, for each $a\in A^{\prime}$, choose two distinct values $d_a^-,d_a^+\in D$, with $d_a^-<d_a^+$,
such that
\[
q\in I_{a,d_a^-}\cap I_{a,d_a^+}.
\]
Such a choice is possible because, for every $a\in A^{\prime}$, there are uncountably many values of
$d\in D$ for which $q\in I_{a,d}$. 
Notice that the rational number $q$ is the same for every $a\in A^{\prime}$, although the two values $d_a^-$ and $d_a^+$ may depend on $a$.

Let $a,a^{\prime}\in A^{\prime}$ with $a<a^{\prime}$.
We must have $d_a^+\leq d_{a^{\prime}}^-$.
Otherwise, the claim in the preceding paragraph would imply
\[
I_{a,d_a^+}\cap I_{a^{\prime},d_{a^{\prime}}^-}=\varnothing,
\]
contradicting the fact that both intervals contain $q$.
Hence the non-empty intervals
\[
(d_a^-,d_a^+),\; a\in A^{\prime},
\]
are pairwise disjoint.
This is impossible, since $A^{\prime}$ is uncountable and every non-empty open interval in $\R$ contains a rational number.
Therefore, no real-valued SWF satisfying Strong Equity exists when $Y$ is uncountable.
\end{proof}

Propositions \ref{P1} and \ref{P2} yield the following characterization of the domains admitting a real-valued SWF satisfying Strong Equity.

\begin{theorem}
\label{T3}
\emph{Let $Y\subseteq\R$ be non-empty, and let $X\equiv Y^\N$.
There exists a social welfare function $W:X\to\R$ satisfying Strong Equity if and only if $Y$ is countable.}
\end{theorem}

\begin{proof}
If $Y$ is countable, Proposition \ref{P1} establishes the existence of a social welfare function satisfying Strong Equity.
If $Y$ is uncountable, Proposition \ref{P2} shows that such a function cannot exist.
\end{proof}

Theorem \ref{T3} provides a complete characterization of the domains on which a real-valued SWF satisfying Strong Equity exists.
However, the existence of an SWF does not imply that it admits an explicit description.
Further, the question of whether such an SWF can be continuous is distinct from the question of its existence.
We examine these issues next.

\subsection{Construction and Non-constructiveness of SWFs}
Observe that the proof of Proposition \ref{P1} is not constructive in the sense used in this paper.
The appeal to the Szpilrajn Lemma is an appeal to an order-extension principle, customarily proven by means of Zorn's Lemma, and the proof also requires choosing a linear extension $<_E$ and a representing map $e_E$ on each eventual-equivalence class.
Thus Proposition \ref{P1} does not itself provide an explicit description of an SWF satisfying SE for general countable domains.
We first consider a seven-element domain and provide an explicitly defined SWF satisfying Strong Equity.
The function is also continuous in the standard Baire topology.
We then show that, when the domain contains eight or more elements, the existence of an SWF satisfying Strong Equity implies the existence of a non-Ramsey set.
\begin{example}\label{Ex1}
Let $Y=\{a,b,c,d,e,f,g\}$ where $a<b<c<d<e<f<g$.
Given $x\in X$, define
\[
\begin{cases}
N_0(x)&\equiv \{n\in\N:x_n=a \text{ or } x_n=g\},\\[6pt]
N_1(x)&\equiv \{n\in\N:x_n=b \text{ or } x_n=f\},\\[6pt]
N_2(x)&\equiv \{n\in\N:x_n=c \text{ or } x_n=d  \text{ or } x_n=e\},
\end{cases}
\]
and 
\[
\alpha_0(n) = \frac{1}{5^n},\;\alpha_1(n) = \frac{1}{3^n},\;\alpha_2(n)= \frac{1}{2^n}.
\]
Define $W:X\to\R$ by
\[
W(x)= \sum_{n\in N_0(x)}\alpha_0(n)+ \sum_{n\in N_1(x)}\alpha_1(n) + \sum_{n\in N_2(x)}\alpha_2(n).
\]
Then $W$ is continuous in the standard Baire topology and satisfies Strong Equity.%
\footnote{An example of similar nature was communicated to the first author via an unpublished note by  \citet{mitra2010}.
It was referred to in \citet{dubey2016} without specifying the details.}

\end{example}

\begin{proof}
Let $x,y\in X$ be such that there exist $i,j\in\N$ with $y_j>x_j>x_i>y_i$ and $x_k=y_k$ for all  $k\in\N\setminus\{i,j\}$.
We need to show that $W(x)>W(y)$.  
Since $x$ and $y$ differ only at coordinates $i$ and $j$, it is enough to compare the contributions of these two coordinates.

Note first that $x_j\in\{c,d,e,f\}$.  
So, the contribution of $x_j$ to $W(x)$, call it $U(x)$, is either $\alpha_2(j)$ or $\alpha_1(j)$.  
If it is $\alpha_1(j)$, then $y_j=g$ and the contribution of $y_j$ to $W(y)$, call it $U(y)$, is $\alpha_0(j)<\alpha_1(j)$.  
If $U(x)=\alpha_2(j)$, then $U(y)\leq U(x)$, and equality can occur only if $x_j,y_j\in\{c,d,e\}$.  
Since $y_j>x_j$, in this equality case we must have $x_j\in\{c,d\}$.

Note next that $x_i\in\{b,c,d,e\}$.  
So, the contribution of $x_i$ to $W(x)$, call it $V(x)$, is either $\alpha_1(i)$ or $\alpha_2(i)$.  
If it is $\alpha_1(i)$, then $y_i=a$ and the contribution of $y_i$ to $W(y)$, call it $V(y)$, is $\alpha_0(i)<\alpha_1(i)$.  
If $V(x)=\alpha_2(i)$, then $V(y)\leq V(x)$, and equality can occur only if $x_i,y_i\in\{c,d,e\}$.  
Since $x_i>y_i$, in this equality case we must have $x_i\in\{d,e\}$.

If neither of the two inequalities above is strict, then we must have $x_j\in\{c,d\}$ and $x_i\in\{d,e\}$.  
But then $x_i\geq d\geq x_j$, contradicting $x_j>x_i$.  
Hence at least one of the two inequalities is strict.  
Since all other coordinates have the same contribution in $W(x)$ and $W(y)$, we get $W(x)>W(y)$.
Thus $W$ satisfies Strong Equity.

Next, we show that \(W\) is continuous in the standard Baire topology.
Recall that the Baire topology is generated by the basic open sets
\(U_{\sigma}\), for \(\sigma \in Y^{<\omega}\), defined by
\[
U_{\sigma}
:=
\left\{
y \in Y^{\N} : \sigma \subseteq y
\right\}.
\]
Fix \(\alpha \in (0,1)\) and let \(\varepsilon>0\) be arbitrary. 
Define
\[ 
X_{\alpha} =\left\{ z\in Y^{\N}: W(z)= \alpha\right\},
\] 
and
\[
X_{\alpha,\varepsilon}
:=
\left\{
z \in Y^{\N} :
W(z) \in (\alpha-\varepsilon,\alpha+\varepsilon)
\right\}.
\]
Let \(z \in X_{\alpha,\varepsilon}\), and set
\[
\delta
:=
\varepsilon-\lvert W(z)-\alpha\rvert>0.
\]
For every \(m \in \omega\), note that
\[
\sup
\left\{
\lvert W(x)-W(z)\rvert :
x \in U_{z\upharpoonright m}
\right\}
\leq
\sum_{k=m}^{\infty}\frac{1}{2^k}
=
\frac{1}{2^{m-1}}.
\]
Hence, we may choose \(m \in \omega\) sufficiently large that
\[
\frac{1}{2^{m-1}}<\delta.
\]
It follows that
\[
U_{z\upharpoonright m}
\subseteq
X_{\alpha,\varepsilon}.
\]
Therefore, \(X_{\alpha,\varepsilon}\) is open, and thus \(W\) is continuous.
\end{proof}

We now establish the non-Ramsey implication for an eight-element domain.

\begin{proposition}
\label{P21}
\emph{Let
\[
Y=\{a,b,c,d,e,f,g,h\},\; \text{such that}\; a<b<c<d<e<f<g<h,
\]
and let $X=Y^{\N}$.
Then the existence of a social welfare function \(W:X\to\mathbb{R}\) satisfying Strong Equity implies the existence of a non-Ramsey set.}
\end{proposition}

\begin{proof}
Let
\[
\mathcal{S}=[\mathbb{N}]^{\mathbb{N}},
\]
and consider \(\alpha\in\mathcal{S}\).
Denote
\[
\alpha=\{t_1,t_2,\cdots\},
\]
where \(t_k<t_{k+1}\) for every \(k\in\mathbb{N}\).
Given \(\alpha\), we construct a pair of infinite utility streams as
follows:
\[
x_{\alpha}(n)=
\begin{cases}
c, & n=1,\\
f, & n=2,\\
b, & n=2m+1,\; m\in\alpha,\\
g, & n=2m+2,\; m\in\alpha,\\
a, & n=2m+1,\; m\notin\alpha,\\
h, & n=2m+2,\; m\notin\alpha,
\end{cases}
\]
and
\[
y_{\alpha}(n)=
\begin{cases}
d, & n=1,\\
e, & n=2,\\
x_{\alpha}(n)& n>2.
\end{cases}
\]
Only the first two terms differ between \(x_{\alpha}\) and
\(y_{\alpha}\). For every \(t_k\in\alpha\), the terms \(2t_k+1\) and
\(2t_k+2\) are assigned \(b\) and \(g\), respectively. For every
\(m\in\mathbb{N}\setminus\alpha\), the terms \(2m+1\) and \(2m+2\) are
assigned \(a\) and \(h\), respectively. Observe that
\[
x_{\alpha}(1)=c<d=y_{\alpha}(1)
<y_{\alpha}(2)=e<f=x_{\alpha}(2).
\]
Therefore, by Strong Equity,
\[
W(x_{\alpha})<W(y_{\alpha}).
\]
Further, let
\[
\alpha^{\prime}=\alpha\setminus\{t_k\}
\]
for some \(k\in\mathbb{N}\). Then
\begin{center}
\renewcommand{\arraystretch}{1.3}
\begin{tabular}{c|cccc}
 & \(1\) & \(2\) & \(2t_k+1\) & \(2t_k+2\)\\
\hline
\(y_{\alpha^{\prime}}\) & \(d\) & \(e\) & \(a\) & \(h\)\\
\(x_{\alpha}\)          & \(c\) & \(f\) & \(b\) & \(g\)
\end{tabular}
\end{center}
There are no changes in the utilities of any other generation.
Observe that
\[
y_{\alpha^{\prime}}(2t_k+1)
=a
<b
=x_{\alpha}(2t_k+1)
<c
=x_{\alpha}(1)
<d
=y_{\alpha^{\prime}}(1),\;\text{and}
\]
\[
y_{\alpha^{\prime}}(2)
=e
<f
=x_{\alpha}(2)
<g
=x_{\alpha}(2t_k+2)
<h
=y_{\alpha^{\prime}}(2t_k+2).
\]
\noindent Let \(z\in X\) be defined as follows:
\[
z(n)=
\begin{cases}
c, & n=1,\\
b, & n=2t_k+1,\\
y_{\alpha^{\prime}}(n), & \text{otherwise}.
\end{cases}
\]
By the first of the preceding inequalities and Strong Equity,
\[
W(y_{\alpha^{\prime}})<W(z).
\]
Moreover, \(z\) and \(x_{\alpha}\) differ only at generations \(2\) and
\(2t_k+2\). Therefore, by the second inequality and Strong Equity,
\[
W(z)<W(x_{\alpha}).
\]
Hence,
\[
W(y_{\alpha^{\prime}})<W(x_{\alpha}).
\]
Thus, we have obtained the nonempty open interval
\[
\mathbb{I}(\alpha)
=
\bigl(W(x_{\alpha}),W(y_{\alpha})\bigr).
\]
For every
\[
\alpha^{\prime}=\alpha\setminus\{t_k\},
\;
t_k\in\alpha,
\]
we have
\[
W(y_{\alpha^{\prime}})<W(x_{\alpha}).
\]
Since Strong Equity also gives
\[
W(x_{\alpha^{\prime}})<W(y_{\alpha^{\prime}}),
\]
we obtain
\[
W(x_{\alpha^{\prime}})
<
W(y_{\alpha^{\prime}})
<
W(x_{\alpha})
<
W(y_{\alpha}).
\]
Consequently,
\[
\mathbb{I}(\alpha^{\prime})
=
\bigl(W(x_{\alpha^{\prime}}),W(y_{\alpha^{\prime}})\bigr)
\]
lies entirely to the left of \(\mathbb{I}(\alpha)\). In particular,
\[
\mathbb{I}(\alpha^{\prime})\cap I(\alpha)=\varnothing.
\]
Next, take an enumeration of the rational numbers and keep it fixed
for the remainder of the proof:
\[
\mathbb{Q}=\{q_1,q_2,\cdots\}.
\]
Since every nonempty open interval contains a rational number, let
\[
\gamma(\alpha)
=
\min\left\{
n\in\mathbb{N}:q_n\in I(\alpha)
\right\}.
\]
Since \(I(\alpha^{\prime})\) and \(I(\alpha)\) are disjoint, we have
\[
\gamma(\alpha)\neq\gamma(\alpha^{\prime})
\]
whenever
\[
\alpha^{\prime}=\alpha\setminus\{t_k\},
\;
t_k\in\alpha.
\]
Now consider the following collection of infinite subsets of
\(\mathbb{N}\):
\[
\Omega
=
\left\{
\alpha\in\mathcal{S}:
\gamma(\alpha)
<
\gamma\bigl(\alpha\setminus\{n\}\bigr)
\text{ for every }n\in\alpha
\right\}.
\]
We claim that \(\Omega\) is non-Ramsey.
Let
\[
\alpha=\{\alpha_1,\alpha_2,\cdots\}\in\mathcal{S},
\]
where \(\alpha_i<\alpha_{i+1}\) for every \(i\in\mathbb{N}\). We first
show that
\[
[\alpha]^{\mathbb{N}}\nsubseteq\Omega.
\]
Suppose, to the contrary, that
\[
\beta\in\Omega
\;
\text{for every }\beta\in[\alpha]^{\mathbb{N}}.
\]
Consider
\[
\beta=\{\alpha_2,\alpha_4,\cdots\}.
\]
Then \(\beta\in[\alpha]^{\mathbb{N}}\). For every
\(m\in\mathbb{N}\), define
\[
\beta_m
=
\beta\cup
\{\alpha_1,\alpha_3,\cdots,\alpha_{2m-1}\}.
\]
Thus,
\[
\beta_1=\beta\cup\{\alpha_1\},
\;
\beta_2=\beta\cup\{\alpha_1,\alpha_3\},
\]
and so on. Then \(\beta_m\in[\alpha]^{\mathbb{N}}\) for every
\(m\in\mathbb{N}\), and
\[
\beta_m
=
\beta_{m+1}\setminus\{\alpha_{2m+1}\}.
\]
By our supposition, \(\beta_{m+1}\in\Omega\). It follows from the
definition of \(\Omega\) that
\[
\gamma(\beta_{m+1})<\gamma(\beta_m)
\]
for every \(m\in\mathbb{N}\). This gives an infinite strictly
decreasing sequence in \(\mathbb{N}\), which is impossible. Therefore,
there exists
\[
\beta^{\prime}\in[\alpha]^{\mathbb{N}}
\]
such that
\[
\beta^{\prime}\notin\Omega.
\]
We next show that
\[
[\alpha]^{\mathbb{N}}\cap\Omega\neq\varnothing.
\]
Suppose, to the contrary, that
\[
\beta\notin\Omega
\;
\text{for every }\beta\in[\alpha]^{\mathbb{N}}.
\]
Set
\[
\beta_0=\alpha.
\]
Since \(\beta_0\notin\Omega\), there exists
\(t^{\prime}_1\in\beta_0\) such that
\[
\gamma(\beta_0)
\not<
\gamma\bigl(\beta_0\setminus\{t^{\prime}_1\}\bigr).
\]
The two values are distinct, and hence
\[
\gamma(\beta_0)
>
\gamma\bigl(\beta_0\setminus\{t^{\prime}_1\}\bigr).
\]
Denote
\[
\beta_1=\beta_0\setminus\{t^{\prime}_1\}.
\]
Then \(\beta_1\in[\alpha]^{\mathbb{N}}\), and hence
\(\beta_1\notin\Omega\) by assumption. Therefore, there exists
\(t^{\prime}_2\in\beta_1\) such that
\[
\gamma(\beta_1)
>
\gamma\bigl(\beta_1\setminus\{t^{\prime}_2\}\bigr).
\]
Denote
\[
\beta_2=\beta_1\setminus\{t^{\prime}_2\}.
\]
Continuing inductively, since every \(\beta_m\) is infinite, we obtain
\[
\gamma(\alpha)
=
\gamma(\beta_0)
>
\gamma(\beta_1)
>
\gamma(\beta_2)
>
\cdots,
\]
which is an infinite strictly decreasing sequence in
\(\mathbb{N}\). This is impossible. Hence, there exists some
\[
\beta^{\prime\prime}\in[\alpha]^{\mathbb{N}}
\]
such that
\[
\beta^{\prime\prime}\in\Omega.
\]

Since \(\alpha\in\mathcal{S}\) was arbitrary, every set
\([\alpha]^{\mathbb{N}}\) contains an element of \(\Omega\) and an
element of \(\mathcal{S}\setminus\Omega\). Therefore, \(\Omega\) is
non-Ramsey.
\end{proof}

Example \ref{Ex1} and Proposition \ref{P21} depict the following picture.
For $Y$ limited to seven distinct elements we have an SWF satisfying SE which is defined explicitly and is also continuous in the standard product topology.
However, for $Y$ containing more than seven elements, the existence of an SWF satisfying SE entails existence of a non-Ramsey set, which is known to be a non-constructive object.
Thus Proposition \ref{P1} is an existence result and does not itself yield an explicit description of the SWF for general countable domains.

This refines \citet[Proposition 3]{sakamoto2012}  where it has been shown that there exist SWF satisfying PD and Weak Dominance (WD) for $Y=[0,1]$. %
\footnote{Weak Dominance is defined as follows: for all $x,y\in X$, if there exists $j\in\N$ such that $x_j>y_j$ and $x_n=y_n$ for all $n\neq j$, then $W(x)>W(y)$.}
Our result on SWF satisfying SE is applicable here and shows that the SWF satisfying PD only is non-constructive for $Y$ containing more than seven elements.
Hence, WD is not needed to prove the non-constructive nature here.
Further,  \citet[Propostion 5]{alcantud2010a} prove that there exist SWF satisfying PD, AN and WD for $Y=[0,1]$ and \citet[Theorem 5]{mitra2007c} prove that there exist SWF satisfying AN and WD for $Y=[0,1]$.
We refine these two results in the following manner.
Using the proof technique adopted here, it is possible to show that for $Y=\{a, b\}$, existence of any SWF satisfying AN and WD implies the existence of a non-Ramsey set.
We also show that SWF satisfying AN and PD is non-constructive in nature for and $Y$ containing four or more elements, i.e., for all non-trivial domain $Y$.
We provide the details in the appendix.

\subsection{Continuity of SWFs}
The preceding results concern the existence and explicit construction of SWFs satisfying Strong Equity.
We now consider whether such an SWF can be continuous in the standard Baire topology.
Example \ref{Ex1} provides a positive result for a seven-element domain.
The following proposition establishes a negative result for an eight-element domain.%
\footnote{The restriction to the codomain $[0,1]$ entails no loss of generality for continuous real-valued SWFs, since composition with a strictly increasing continuous function from $\R$ into $(0,1)$ preserves continuity and Strong Equity.}

\begin{proposition}\label{Pbaire}
\emph{Let
\[
Y=\{a,b,c,d,e,f,g,h\},\; \text{such that}\; a<b<c<d<e<f<g<h,
\]
and let $X=Y^{\N}$.
Suppose
\[
W:X\longrightarrow [0,1]
\]
is a utility function satisfying Strong Equity.
Then $W$ cannot be continuous with respect to the Baire topology on $Y^{\N}$}.
\end{proposition}

\begin{proof}
Let $(q_n)_{n\geq 1}$ be an enumeration of all rational numbers in $(0,1)$.
Fix $\alpha\in(0,1)$ and define
\[
L(\alpha)=\{m:q_m<\alpha\},\; \text{and}\; 
U(\alpha)=\{m:q_m\geq\alpha\}.
\]
Define the sequence $x_\alpha$ and $y_{\alpha}$ by
\[
x_\alpha(m)=
\begin{cases}
c, & m=1,\\[1mm]
f, & m=2,\\[1mm]
a, & m=2n+1,\; n\in U(\alpha),\\[1mm]
h, & m=2n+2,\; n\in U(\alpha),\\[1mm]
b, & m=2n+1,\; n\in L(\alpha),\\[1mm]
g, & m=2n+2,\; n\in L(\alpha), 
\end{cases}
\]
and
\[
y_\alpha(m)=
\begin{cases}
d, & m=1,\\
e, & m=2,\\
x_\alpha(m), & m>2,
\end{cases}
\]
respectively.
Note that
\[
x_\alpha(1)=c
<
y_\alpha(1)=d
<
y_\alpha(2)=e
<
x_\alpha(2)=f.
\]
By SE,
\[
y_\alpha\succ x_\alpha, \; \text{or}\; 
W(y_\alpha)>W(x_\alpha).
\]
Now pick one
\[
\bar m\in U(\alpha)
\]
and switch the pair $(a,h)$ with the pair $(b,g)$ at coordinates $2\bar m+1$, and $2\bar m +2$ to  obtaining a sequence $z$. Then
\[
W(z)>W(y_\alpha).
\]
Indeed,
\[
\begin{array}{c|cc|cc}
&1&2&2\bar m+1&2\bar m+2\\ \hline
x_\alpha&c&f&a&h\\[1mm]
z&c&f&b&g\\[1mm]
y_\alpha&d&e&a&h
\end{array}
\]
We have
\[
y_\alpha(1)=d
>
z(1)=c
>
z(2\bar m+1)=b
>
y_\alpha(2\bar m+1)=a,
\]
and
\[
y_\alpha(2)=e
<
z(2)=f
<
z(2\bar m+2)=g
<
y_\alpha(2\bar m+2)=h.
\]
Both pairs are ranked by SE, and therefore
\[
z\succ y_\alpha,\; \text{or}\; 
W(z)>W(y_\alpha).
\]
Now let $\beta>\alpha$ and define
\[
I=L(\beta)\setminus L(\alpha)
=\{m_0,m_1,m_2,\cdots\}.
\]
Observe that $I\subseteq U(\alpha)$.
Let $z_0$ be obtained from $x_\alpha$ by switching the pair $(a,h)$ with
$(b,g)$ at coordinates $2m_0+1$, and $2m_0+2$.
By the preceding argument,
\[
W(z_0)>W(y_\alpha)>W(x_\alpha).
\]
Inductively, define $z_{k+1}$ from $z_k$ by switching the pair $(a,h)$
with $(b,g)$ at coordinates $2m_{k+1}+1$, and $2m_{k+1}+2$.
Then, by SE, $W(z_{k+1})>W(z_k)$.
Hence we obtain
\[
W(x_\alpha)
<
W(y_\alpha)
<
W(z_0)
<
W(z_1)
<
\cdots
<
W(z_n)
<
W(z_{n+1})
<
\cdots.
\]
Moreover, in the standard Baire topology,
\[
z_k\longrightarrow x_\beta.
\]
By continuity, for every $k\geq 0$,
\[
W(x_\beta)
=
\lim_{j\to\infty}W(z_j)
\geq
W(z_k)
>
W(y_\alpha).
\]
Thus, for each $\alpha\in(0,1)$, define the nonempty open interval
\[
\mathbb{I}_\alpha=(W(x_\alpha),W(y_\alpha)).
\]
For $\alpha<\beta$, we have
\[
W(y_\alpha)<W(x_\beta),
\]
and therefore
\[
\mathbb{I}_\alpha\cap \mathbb{I}_\beta=\varnothing.
\]
Moreover,
\[
r\in \mathbb{I}_\alpha,\;
s\in \mathbb{I}_\beta
\;\Longrightarrow\;
r<s.
\]
This yields uncountably many pairwise disjoint open intervals in
$\mathbb{R}$, which is impossible since every open interval contains a
rational number and the rationals are countable.
Hence $W$ cannot be continuous with respect to the Baire topology on
$Y^\N$.
\end{proof}

\subsection{Construction of SWOs}
The preceding results concern real-valued social welfare functions satisfying Strong Equity.
We now turn to the explicit construction of social welfare orders.
Unlike a real-valued representation, an explicitly defined SWO need not assign real numbers to utility streams.
We identify restrictions on the domain under which such an SWO can be constructed.

\begin{proposition}
\label{P3}
\emph{Let $Y\subseteq\R$ be non-empty and let $X\equiv Y^\N$.
Let
\[
L(Y)\equiv \{a\in Y: \exists\; b , c, d \in Y \text{ with } a<b<c<d\},\; \text{and}
\]
\[
U(Y)\equiv \{d\in Y: \exists\;  a, b, c \in Y \text{ with } a<b<c<d\}.%
\footnote{Note that \(L(Y)\) is the set of utility levels that can occur as the pre-transfer poor generation and \(U(Y)\) is the set of utility levels that can occur as the pre-transfer rich generation in a Strong Equity comparison.}
\]
\[
\text{If either}\; L(Y)\;\text{is well-ordered by the usual order}\;\text{or}\; U(Y)\; \text{ is reverse well-ordered},
\]
then there exists an explicitly defined SWO on $X$ satisfying Strong Equity.}
\end{proposition}

\begin{proof}
We prove the two cases separately.
Observe that $|Y|\geq 4$ for non-trivial SE rankings.

\noindent\textbf{$L(Y)$ is well-ordered.}
For $x\in X$, $t\in L(Y)$, and $n\in\N$, we define
\[
B_{t,n}(x)=
\begin{cases}
1, & x_n\leq t,\\
0, & x_n>t.
\end{cases}
\]
We order $L(Y)\times\N$ lexicographically, threshold utility $t$ first and generation $n$ second.  
Since $L(Y)$ is well-ordered, this index set is well-ordered.  
We rank utility streams lexicographically by their $B$-codes, with $0$ ranked above $1$.  
If all $B$-codes agree, we complete the ranking as follows.
If \(x=y\), declare \(x\sim y\).
If \(x\neq y\), let \(n\) be the first generation for which \(x_n\neq y_n\). 
We rank \(x\succ y\) if \(x_n>y_n\), and rank \(y\succ x\) otherwise.
This rule is complete and transitive.  
Indeed, either there is a first $B$-code difference, or all $B$-codes agree and there is a first generation at which the streams differ, or the streams are identical.  
Transitivity follows from the standard transitivity of lexicographic ranking over a well-ordered index set, followed by a lexicographic tie-breaker.

Now let $x,y\in X$ be such that there exist $i,j\in\N$ with $y_j>x_j>x_i>y_i$ and $x_k=y_k$ for every $k\in\N\setminus\{i,j\}$.
Observe that $y_i\in L(Y)$.  
For every $s\in L(Y)$ with $s<y_i$, both affected coordinates are above $s$ in $x$ and $y$, so no $B$-code difference appears before the threshold utility $y_i$.  
At threshold utility $y_i$, coordinate $i$ gives
\[
B_{y_i,i}(x)=0,\; B_{y_i,i}(y)=1.
\]
At the same threshold utility $y_i$, coordinate $j$ gives
\[
B_{y_i,j}(x)=B_{y_i,j}(y)=0.
\]
All other coordinates are identical in the two streams.  
Hence the constructed order ranks $x\succ y$, as required by SE.

\noindent\textbf{$U(Y)$ is reverse well-ordered.}
For $x\in X$, $t\in U(Y)$, and $n\in\N$, define
\[
C_{t,n}(x)=
\begin{cases}
1, & x_n<t,\\
0, & x_n\geq t.
\end{cases}
\]
We order $U(Y)\times\N$ lexicographically with threshold utilities ordered from larger to smaller, and generations ordered increasingly.  
Since $U(Y)$ is reverse well-ordered, this index set is well-ordered.  
We rank utility streams lexicographically by their $C$-codes, with $1$ ranked above $0$.  
If all $C$-codes agree, we complete the ranking as follows.
If \(x=y\), declare \(x\sim y\).
If \(x\neq y\), let \(n\) be the first generation for which \(x_n\neq y_n\). 
We rank \(x\succ y\) if \(x_n>y_n\), and rank \(y\succ x\) otherwise.
The same lexicographic argument gives a complete and transitive order.  
Now let $x,y\in X$ be such that there exist $i,j\in\N$ with $y_j>x_j>x_i>y_i$ and $x_k=y_k$ for every $k\in\N\setminus\{i,j\}$.  
Observe that $y_j\in U(Y)$.  
For every $s\in U(Y)$ with $s>y_j$, both affected coordinates are below $s$ in $x$ and $y$, so no $C$-code difference appears before threshold utility $y_j$.  
At threshold utility $y_j$, coordinate $j$ gives
\[
C_{y_j,j}(x)=1,\; C_{y_j,j}(y)=0.
\]
At the same threshold, coordinate $i$ gives
\[
C_{y_j,i}(x)=C_{y_j,i}(y)=1.
\]
All other coordinates are identical. 
Hence the constructed order ranks $x\succ y$, as required by SE.
\end{proof}

\subsection{Non-constructive SWO}
\label{S1}
In this sub-section we show that, for any domain of order type $\mu$, the existence of an SWO satisfying SE implies the existence of a non-Ramsey set.

\begin{proposition}
\label{P4}
\emph{Let $Y\subseteq\R$ be of order type $\mu$, and  let $X\equiv Y^\N$.  
If there exists a social welfare order on $X$ satisfying Strong Equity, then there exists a non-Ramsey collection $\OmegaC\subseteq\calS$.}
\end{proposition}
\noindent The following lemma is useful in the proof of the proposition.

\begin{lemma}\label{L4}
\emph{Let \(z=(z_k)_{k\in\Z}\) be strictly increasing.
For every \(n\in\N\), define the two-coordinate blocks $A_n=(z_{-2n},z_{2n+1})$, $B_n=(z_{-2n+1},z_{2n})$.
Let \(F(R)\in \{z_k:k\in\Z\}^{\N}\) be the infinite utility sequence obtained from \(R\subseteq\N\) by assigning block \(B_n\) to coordinates \((2n-1,2n)\) if \(n\in R\), and block \(A_n\) to coordinates \((2n-1,2n)\) if \(n\notin R\).
Let \(R,Q\subseteq\N\). If \(Q\setminus R\) and \(R\setminus Q\) are finite, \(Q\neq R\), and there exists an injection $\varphi:R\setminus Q\to Q\setminus R$ such that \(\varphi(n)>n\) for every \(n\in R\setminus Q\),%
\footnote{For example, take \(z_k=k\). 
Let $R=\{1,5,7,9,\cdots\}$, $Q=\{3,5,7,9,\cdots\}$.
Then \(R\setminus Q=\{1\}\), \(Q\setminus R=\{3\}\), and the map \(\varphi(1)=3\) satisfies \(\varphi(1)>1\). 
The two streams agree on all blocks except those indexed by \(1\) and \(3\). 
In block notation, $F(R)=(B_1,A_2,A_3,A_4,B_5,A_6,B_7,\cdots)$, whereas $F(Q)=(A_1,A_2,B_3,A_4,B_5,A_6,B_7,\cdots)$.
Equivalently, $F(R)=(-1,2,-4,5,-6,7,\cdots)$, while $F(Q)=(-2,3,-4,5,-5,6,\cdots)$.
Here \(A_1=(-2,3)\), \(B_1=(-1,2)\), \(A_3=(-6,7)\), and \(B_3=(-5,6)\).}
then every SWO satisfying Strong Equity ranks \(F(Q)\succ F(R)\).}
\end{lemma}

\begin{proof}
First observe that, for each $n\in\N$, $z_{-2n}<z_{-2n+1}<z_{2n}<z_{2n+1}$.
Hence replacing block $A_n=(z_{-2n},z_{2n+1})$ by block $B_n=(z_{-2n+1},z_{2n})$ is one Strong Equity improvement.
Next fix $n<m$.  
Consider the two-block replacement $(B_n, A_m) = (z_{-2n+1},z_{2n}), (z_{-2m},z_{2m+1})$ $\longmapsto$ $(z_{-2n}, z_{2n+1}),  (z_{-2m+1} ,z_{2m})= (A_n, B_m)$.
This replacement is a finite chain of two Strong Equity improvements.  
Note that $(z_{-2n+1},z_{-2m})$ is replaced by $(z_{-2n},z_{-2m+1})$.
Since $n<m$, $z_{-2n+1}>z_{-2n}>z_{-2m+1}>z_{-2m}$, so this is a Strong Equity improvement.  
The remaining pair of coordinates $(z_{2n},z_{2m+1})$ are replaced by $(z_{2n+1},z_{2m})$.
Since $n<m$, $z_{2m+1}>z_{2m}>z_{2n+1}>z_{2n}$, so this is also a Strong Equity improvement.  
Thus the two-block replacement $B_n, A_m\mapsto A_n,B_m$ is forced by Strong Equity and transitivity.

Now compare $F(R)$ and $F(Q)$.  
For every $n\in R\setminus Q$, pair the loss of block $B_n$ with the gain of block $B_{\varphi(n)}$, where $\varphi(n)\in Q\setminus R$ and $\varphi(n)>n$.  
The preceding paragraph shows that each paired replacement is a finite chain of Strong Equity improvements.  
Since $\varphi$ is injective, these paired replacements use disjoint gained blocks.  
Any remaining elements of $Q\setminus R$ are pure replacements of $A_m$ by $B_m$, and each such replacement is itself a Strong Equity improvement.  
There are only finitely many changed blocks, so applying these improvements successively yields a finite strict-improvement chain from $F(R)$ to $F(Q)$.  
Therefore $F(Q)\succ F(R)$.
\end{proof}

\begin{proof}[Proof of Proposition~\ref{P4}]
Let $\succeq$ be an SWO on $X$ satisfying Strong Equity.  
Choose a strictly increasing family $\{z_k:k\in\mathbb Z\}\subseteq Y$.  
It is enough to work with streams whose coordinates belong to $Z=\{z_k:k\in\mathbb Z\}$, because the restriction of $\succeq$ to $Z^\N$ is still complete, transitive, and  satisfies Strong Equity.
For each $n\in\N$, define the two-coordinate blocks $A_n=(z_{-2n}, z_{2n+1})$, $ B_n=(z_{-2n+1}, z_{2n})$, and define $F(R)$ for $R\subseteq\N$ as in Lemma \ref{L4}.

For every $T\in\calS$, we write $T=\{t_1<t_2<t_3<\cdots\}$, and define the finite non-empty intervals $I_k(T)=\{n\in\N:t_k\leq n<t_{k+1}\}\; (k\in\N)$.
We define the odd and even interval unions 
\[
O(T)=\bigcup_{k=0}^{\infty} I_{2k+1}(T),\;
E(T)=\bigcup_{k=1}^{\infty} I_{2k}(T).
\]
Let 
\[
x(T)=F(O(T)),\; y(T)=F(E(T)),
\]
and define
\[
\OmegaC=\{T\in\calS:y(T)\succ x(T)\}.
\]

We show that $\OmegaC$ is non-Ramsey.  
Let  $T\in\calS$ be arbitrary.
There are two cases.
\begin{enumerate}
\item
$T\in\OmegaC$.  
We drop $t_1$ to obtain $S=\{t_2,t_3,t_4,\cdots\}$.
Then $S\in\calS(T)$.  
The interval decomposition of $S$ satisfies $O(S)=E(T)$, $E(S)=O(T)\setminus I_1(T)$.
Therefore, $x(S)=F(O(S))=F(E(T))=y(T)$.
Also, $x(T)=F(O(T))$ and $y(S)=F(O(T)\setminus I_1(T))$.  
Since $I_1(T)$ is finite and non-empty, Lemma \ref{L4}, applied with $R=O(T)\setminus I_1(T)$ and $Q=O(T)$, gives $x(T)\succ y(S)$.
Thus, $x(S)=y(T)\succ x(T)\succ y(S)$, so $x(S)\succ y(S)$ by transitivity.  
Hence $S\notin\OmegaC$.  

\item
$T\notin\OmegaC$.  
Since  $\succeq$ is complete, we get $x(T)\succeq y(T)$.
Choose an odd integer $p\geq 3$. 
Choose an even integer $q>p+1$ so large that the finite set
\[
M=I_{p+1}(T)\cup I_{p+3}(T)\cup\cdots\cup I_{q-2}(T)
\]
has cardinality at least $|I_1(T)|$.  
Such a $q$ exists because every interval $I_k(T)$ is non-empty and there are arbitrarily many even indices after $p$.
Drop $t_1, t_{p+1}, \cdots, t_{q-1}$ to obtain
\[
S=\{t_2,t_3,\cdots,t_p,t_q,t_{q+1},t_{q+2},\cdots\}.
\]
Then $S\in\calS(T)$.
A direct calculation of the interval decomposition of $S$ gives
\[
O(S)=I_2(T)\cup I_4(T)\cup\cdots\cup I_{p-1}(T)
      \cup I_q(T)\cup I_{q+2}(T)\cup\cdots,\; \text{and}
\]
\[
E(S)=I_3(T)\cup I_5(T)\cup\cdots\cup I_{p-2}(T)
      \cup \bigl(I_p(T)\cup I_{p+1}(T)\cup\cdots\cup I_{q-1}(T)\bigr)
      \cup I_{q+1}(T)\cup I_{q+3}(T)\cup\cdots.
\]
Consequently,
\[
O(T)\setminus E(S)=I_1(T), \; \text{and}\; 
E(S)\setminus O(T)=M.
\]
Every element of $M$ is larger than every element of $I_1(T)$, because $p\geq 3$.  Since $|M|\geq |I_1(T)|$, there is an injection $\varphi:I_1(T)\to M$ with $\varphi(n)>n$ for all $n\in I_1(T)$.  
Lemma \ref{L4} applied with $R=O(T)$ and $Q=E(S)$ gives $y(S)=F(E(S))\succ F(O(T))=x(T)$.
Similarly, $E(T)\setminus O(S)=M$, $O(S)\setminus E(T)=\varnothing$.
Since $M$ is finite and non-empty, Lemma \ref{L4} applied with $R=O(S)$ and $Q=E(T)$ gives $y(T)=F(E(T))\succ F(O(S))=x(S)$.
Combining the three comparisons, $y(S)\succ x(T)\succeq y(T)\succ x(S)$, and by  transitivity we get $y(S)\succ x(S)$.
Thus $S\in\OmegaC$.  
\end{enumerate}

In either case, $\calS(T)$ intersects both $\OmegaC$ and $\calS\setminus\OmegaC$.  
Since $T\in\calS$ was arbitrary, $\OmegaC$ is non-Ramsey.
\end{proof}

With the results above, it is useful to summarize some features of the Strong Equity axiom with regard to the existence of SWFs, explicitly constructive SWFs, and SWOs that admit explicit descriptions. 
There are two separate sources of difficulty here. 
The first is the move from an SWO to an SWF: an SWO only has to rank pairs of utility streams, whereas an SWF must represent that ranking by real numbers. 
The second is the move from SE alone to SE together with monotonicity.

Monotonicity is a minimal efficiency requirement, but it is already highly restrictive. 
With SE+M, real-valued representation is known to face severe restrictions on sufficiently rich domains.
In contrast, an explicit SWO can still be constructed on some infinite domains, for example when \(Y\) has order type \(\omega\).
Thus, in the SE $+$ M case, the explicit-SWO result should be distinguished from real-valued SWF representation results.

For SE alone, the relation between the two results is reversed. 
Theorem \ref{T3} shows that a real-valued SWF satisfying SE exists precisely when \(Y\) is countable, while the explicit SWO construction applies only to a smaller class of linearly ordered domains. 
Thus the domains covered by the explicit SWO construction form only a subclass of the countable domains covered by the SWF existence theorem. 
The countable-domain SWF theorem should therefore be read primarily as an existence result, not as an explicit construction. 
This point is especially clear on integer-like domains. 
Using the terminology of \citet{dubey2011}, these are domains of type \(\mu\): \(Y\) contains a subset that is order-isomorphic to the positive and negative integers. 
On such domains, the existence of an SWO satisfying SE entails the existence of a non-Ramsey collection. 
Hence the SWF guaranteed by the countable-domain theorem cannot be constructive on these domains; an explicit SWF would induce an explicit SWO. 
This identifies a clear nonconstructive class for SE, although it should not be read as a complete characterization of all domains on which constructive SE rankings fail.

\begin{remark}\label{R3}
\emph{\begin{itemize}
\item 
Unlike the case of SE on rich domains, there exist real-valued SWFs satisfying PD, as shown in \citet{sakamoto2012} and \citet{alcantud2010a} for $X= [0,1]^{\N}$. 
In addition, the representation satisfies  Weak Dominance.
\item
On domains containing an affine copy of the integers, the existence of an SWO satisfying PD is non-constructive in the same sense as in the SE case. 
The non-Ramsey proof need not be repeated. It is enough to verify the block dominance step. 
For PD, take an affine copy of the integers in \(Y\), say \(z_k=r+k\delta\) with \(\delta>0\), and define $A_n=(z_{-2n}, z_{2n+1})$, $B_n=(z_{-2n+1}, z_{2n})$.
Then $z_{-2n}<z_{-2n+1}<z_{2n}<z_{2n+1}$ and $z_{-2n}+z_{2n+1}=z_{-2n+1}+z_{2n}$.
Thus \(B_n\) is obtained from \(A_n\) by a Pigou-Dalton transfer.
More generally, the same finite-block dominance relation used in the SE proof is generated here by finitely many Pigou-Dalton transfers.
Hence the subsequent construction of a non-Ramsey collection applies without change.
\end{itemize}}
\end{remark}

\section{Conclusions}\label{sec5}

The results of this paper suggest that some correspondence statements in the earlier literature should be reformulated.
Those statements were developed in settings where the domains admitting real-valued social welfare functions appeared to coincide with the domains admitting constructible social welfare orders.
The results obtained here show that this coincidence is not a general feature of consequentialist equity axioms.
For Strong Equity alone, Theorem \ref{T3} establishes that real-valued representation is possible if and only if \(Y\) is countable.
In particular, the domain on which a real-valued SWF satisfying a given equity-efficiency combination exists need not coincide with the domain on which an explicitly constructible SWO satisfying the same axioms exists.

The lexicographic constructions developed in this paper illustrate the distinction between representation and explicit description.
They give explicit SWOs on well-ordered domains even in cases where real-valued representation is unavailable or subject to much stronger restrictions.
Thus restrictions arising in representation theorems should not automatically be interpreted as restrictions on explicit SWO construction.

This observation affects the interpretation of the correspondence principle suggested by earlier work.
For Strong Equity combined with monotonicity, the finite-cardinality boundary in \citet{dubey2014} is relevant for real-valued representation, but not for explicit SWO construction.
A similar qualification applies to Hammond Equity combined with Strong Pareto, as considered in \citet{dubey2014b}, and to the Pigou--Dalton transfer principle, either by itself or together with an efficiency axiom, as considered in \citet{dubey2016}.
For Hammond Equity with Pareto-type efficiency, well-ordered domains, including finite domains and \(\mathbb N\), lie on the constructive side because Theorem \ref{T1} gives explicit SWOs satisfying Strong Pareto and Strong Equity, and hence Hammond Equity.
For Hammond Equity and Weak Pareto, the negative-integers case lies on the non-Ramsey side by \citet[Proposition 1]{dubey2016a}; more generally, the obstruction is generated by descending order structure.
For Pigou--Dalton, the situation is more delicate because the axiom contains an additive conservation condition.
Nonconstructiveness in that setting requires additive richness, not merely a large or non-well-ordered domain.
Theorem \ref{T2}\textup{(b)} gives one such sufficient condition when Pigou--Dalton is combined with monotonicity.

The present results show that the class of domains admitting explicitly constructible SWOs can be larger than the class of domains admitting explicit, or even real-valued, SWF representations.
The appropriate revised formulation is therefore axiom-specific: for each equity and efficiency combination, one must separately identify the representation domain, the explicit-construction domain, and the domain on which non-Ramsey obstructions arise.

\appendix
\section{SWF satisfying AN and WD}

\begin{proposition}
The SWF satisfying AN and WD on $Y=[0,1]$, shown to exist in
\citet[Theorem 5]{mitra2007c}, implies the existence of a non-Ramsey set
for any non-trivial domain $Y$.
\end{proposition}

\begin{proof}
We consider $Y$ to contain two elements $a$ and $b$.
For ease of exposition, take $a=0$ and $b=1$.
Let
\[
\alpha\in[\mathbb{N}]^{\mathbb{N}},
\;
\alpha=(n_1,n_2,n_3,\cdots).
\]
Define the sequences $x_{\alpha}$ and $y_{\alpha}$ by
\[
x_{\alpha}(n)=
\begin{cases}
0, & n=1,\\[1mm]
1, & n=n_k+1 \text{ for some } k\in\mathbb{N},\\[1mm]
0, & \text{otherwise},
\end{cases}
\]
and
\[
y_{\alpha}(n)=
\begin{cases}
1, & n=1,\\[1mm]
x_{\alpha}(n), & n>1,
\end{cases}
\]
respectively.
Since
\[
x_{\alpha}(1)=0<1=y_{\alpha}(1), \;\text{and}\; 
x_{\alpha}(n)=y_{\alpha}(n),
\; n>1,
\]
by WD,
\[
y_{\alpha}\succ x_{\alpha},
\;\text{or}\;
W(y_{\alpha})>W(x_{\alpha}).
\]
Now drop any $n_k\in\alpha$ to obtain
\[
\alpha^{\prime}
=
\alpha\setminus\{n_k\},
\;
\alpha^{\prime}\in[\alpha]^{\mathbb{N}}.
\]
The changes in the relevant coordinates are displayed below:
\[
\begin{array}{c|cc}
&1&n_k+1\\ \hline
x_{\alpha^{\prime}}&0&0\\[1mm]
y_{\alpha^{\prime}}&1&0\\[1mm]
x_{\alpha}&0&1\\[1mm]
y_{\alpha}&1&1
\end{array}
\]
The sequence $y_{\alpha^{\prime}}$ is obtained from $x_{\alpha}$ by
switching coordinates $1$ and $n_k+1$. Therefore, by AN,
\[
x_{\alpha}\sim y_{\alpha^{\prime}},
\;\text{or}\;
W(x_{\alpha})=W(y_{\alpha^{\prime}}).
\]
Combining it with $W(x_{\alpha^{\prime}}) < W(y_{\alpha^{\prime}})$, we obtain
\[
W(x_{\alpha^{\prime}})
<
W(y_{\alpha^{\prime}})
=
W(x_{\alpha})
<
W(y_{\alpha}).
\]
For each $\alpha\in[\mathbb{N}]^{\mathbb{N}}$, define the nonempty open
interval
\[
\mathbb{I}(\alpha)
=
\bigl(W(x_{\alpha}),W(y_{\alpha})\bigr).
\]
Similarly,
\[
\mathbb{I}(\alpha^{\prime})
=
\bigl(
W(x_{\alpha^{\prime}}),
W(y_{\alpha^{\prime}})
\bigr).
\]
Since
\[
W(y_{\alpha^{\prime}})
=
W(x_{\alpha}),
\]
we have
\[
\mathbb{I}(\alpha^{\prime})
\cap
\mathbb{I}(\alpha)
=
\varnothing.
\]
From here on, we can replicate the SE proof to show that any SWF
satisfying AN + WD implies the existence of a non-Ramsey set.
\end{proof}

\section{SWF satisfying AN and PD}

\begin{proposition}
\emph{The existence of an SWF satisfying AN and PD on $Y=[0,1]$, shown in
\citet[Proposition 5]{alcantud2010a}, implies the existence of a
non-Ramsey set for every non-trivial domain $Y$.}
\end{proposition}

\begin{proof}
Let
\[
t\in[\mathbb{N}]^{\mathbb{N}},
\;
t=(t_1,t_2,t_3,\cdots).
\]
Recall, we need four distinct elements in $Y$ to define a PD ranking.
Let $a,b,c,d\in Y$ with $a<b<c<d$ and $a+d=b+c$.
Define the sequence $x_t$ by
\[
x_t(m)=
\begin{cases}
d, & m=1,\\
a, & m=2,\\
c, & m=2n+1,\; n\in t,\\
b, & m=2n+2,\; n\in t,\\
d, & m=2n+1,\; n\notin t,\\
a, & m=2n+2,\; n\notin t,
\end{cases}
\]
and define $y_t$ by
\[
y_t(m)=
\begin{cases}
c, & m=1,\\
b, & m=2,\\
x_t(m), & m>2.
\end{cases}
\]
Thus, $x_t$ is obtained by assigning the pair $(d,a)$ to generations
$1$ and $2$ and, for each $n\in\mathbb{N}$, assigning the pair $(c,b)$
to generations $2n+1$, $2n+2$ if $n\in t$, and assigning the pair
$(d,a)$ to the same generations if $n\notin t$.
The sequence $y_t$ is obtained from $x_t$ by replacing the pair
$(d,a)$ at generations $1$ and $2$ with the pair $(c,b)$. Therefore,
by PD,
\[
y_t\succ x_t,
\;\text{or}\;
W(y_t)>W(x_t).
\]
Now drop any $t_k\in t$ to obtain
\[
t^{\prime}
=
t\setminus\{t_k\},
\;
t^{\prime}\in[t]^{\mathbb{N}}.
\]
The changes in the relevant coordinates are displayed below:
\[
\begin{array}{c|cc|cc}
&1&2&2t_k+1&2t_k+2\\ \hline
x_{t^{\prime}}&d&a&d&a\\[1mm]
y_{t^{\prime}}&c&b&d&a\\[1mm]
x_t(\pi)&c&b&d&a\\[1mm]
x_t&d&a&c&b\\[1mm]
y_t&c&b&c&b
\end{array}
\]
The sequence $x_t(\pi)$ is obtained from $x_t$ by switching
coordinates $1$ and $2t_k+1$, and coordinates $2$ and $2t_k+2$.
Indeed, the pair $(d,a)$ at coordinates $1$ and $2$ in $x_t$ is
switched with the pair $(c,b)$ at coordinates $2t_k+1$ and
$2t_k+2$.
Therefore, by AN,
\[
x_t\sim x_t(\pi),
\;\text{or}\;
W(x_t)=W(x_t(\pi)).
\]
Observe that
\[
y_{t^{\prime}}=x_t(\pi).
\]
Moreover, $y_{t^{\prime}}$ is obtained from $x_{t^{\prime}}$ by
replacing the pair $(d,a)$ at generations $1$ and $2$ with the pair
$(c,b)$. Therefore, by PD,
\[
y_{t^{\prime}}\succ x_{t^{\prime}},
\;\text{or}\;
W(y_{t^{\prime}})>W(x_{t^{\prime}}).
\]
Combining these relations, we obtain
\[
W(x_{t^{\prime}})
<
W(y_{t^{\prime}})
=
W(x_t(\pi))
=
W(x_t)
<
W(y_t).
\]
For each $t\in[\mathbb{N}]^{\mathbb{N}}$, define the nonempty open
interval
\[
\mathbb{I}(t)
=
\bigl(W(x_t),W(y_t)\bigr).
\]
Similarly,
\[
\mathbb{I}(t^{\prime})
=
\bigl(W(x_{t^{\prime}}),W(y_{t^{\prime}})\bigr).
\]
Since
\[
W(y_{t^{\prime}})
=
W(x_t),
\]
we have
\[
\mathbb{I}(t^{\prime})
\cap
\mathbb{I}(t)
=
\varnothing.
\]
From here on, we can replicate the SE proof to show that any SWF
satisfying AN + PD implies the existence of a non-Ramsey set.
\end{proof}

\bibliographystyle{plainnat}
\bibliography{AWPAnonymity}

\end{document}